\documentclass[review,12pt]{elsarticle}
\usepackage{colordvi}

\usepackage{amssymb}
\usepackage{amsthm,amsmath}
\usepackage{mathrsfs}
\usepackage{accents}
\usepackage{graphicx}
\usepackage{epstopdf}
\newtheorem{thm}{Theorem}
\newtheorem{lem}{Lemma}

\newtheorem{asu}{Assumption}
\newtheorem{col}{Corollary}
\newtheorem{rem}{Remark}
\newtheorem{propo}{Proposition}
\newtheorem{nota}{Notation}
\newcommand{\onetom}{1,\cdots,m}

\allowdisplaybreaks

\begin{document}

\begin{frontmatter}

\title{Adaptive algorithms for synchronization, consensus of multi-agents and anti-synchronization of direct complex networks\tnoteref{tt}} \tnotetext[tt]{This work is jointly supported by the National Key $R\&D$ Program of China
(No. 2018AAA010030), National Natural Sciences Foundation of China under
Grant (No. 61673119 and 61673298), STCSM (No. 19JC1420101), Shanghai Municipal Science and Technology Major Project under Grant 2018SHZDZX01 and ZJLab, the Key Project of Shanghai Science and Techonology under Grant 16JC1420402.}

\author[lwl-1,lwl-2,lwl-3,lwl-4]{Wenlian Lu}
\ead{wenlian@fudan.edu.cn}
\author[lxw-1]{Xiwei Liu}
\ead{xwliu@tongji.edu.cn}
\author[ctp-1,lwl-1]{Tianping Chen\corref{tpchen}}
\ead{tchen@fudan.edu.cn}

\address[lwl-1]{School of Mathematical Sciences, Fudan University, Shanghai 200433, China}
\address[lwl-2]{Institute of Science and Technology for Brain-Inspired Intelligence, Fudan University, Shanghai 200433, China} \address[lwl-3]{Shanghai Key Laboratory for Contemporary Applied Mathematics and Laboratory of Mathematics for Nonlinear Science, Fudan University, Shanghai 200433, China}
\address[lwl-4]{Shanghai Center for Mathematical Sciences, Fudan University, Shanghai 200433, China}
\address[lxw-1]{Department of Computer Science and Technology, Tongji University, Shanghai 201804, China}
\address[ctp-1]{School of Computer Science, Fudan University, Shanghai 200433, China}

\cortext[tpchen]{Corresponding author. E-mail address: tchen@fudan.edu.cn}

\begin{abstract}
In this paper, we discuss distributed adaptive algorithms for synchronization of complex networks, consensus of multi-agents with or without pinning controller. The dynamics of individual node is governed by generalized QUAD condition. We design  new algorithms, which can keep the left eigenvector of the adaptive coupling matrix corresponding to the zero eigenvalue invariant. Based on this invariance, various distributive adaptive synchronization, consensus, anti-synchronization models are given.
\end{abstract}


\end{frontmatter}



\section{Introduction}
In more than twenty years, synchronization of complex networks, as a special case, consensus of multi-agents attracted many researchers. The general model is
\begin{equation}\label{model1}
\dot{z}^{i}(t)=\mathscr{F}(z^{i}(t))
+c\sum\limits_{j=1}^{m}l_{ij}\Gamma z^{j}(t),\quad
\end{equation}
where $z^{i}(t)\in \mathbb{R}^{n}$ is the state variable, $i=1,\cdots,m$,
$f:\mathbb{R}^n\rightarrow \mathbb{R}^{n}$, $l_{ij}\ge 0$, for $i\ne j$, and $l_{ii}=-\sum_{j=1}^ml_{ij}$, and $\Gamma\in \mathbb{R}^{n\times n}$.


In \cite{Chen1}, authors proved that network (\ref{model1}) can synchronize if $c$ is large enough. However, large $c$ is impractical. \cite{Chen1b} also provided an example, showing that to reach synchronization, the theoretical value $c_{\min}=0.7$. However, numerical computation shows that if $c>0.06$, the synchronization has reached.
Thus, in \cite{Chen1b} the authors hoped to
find a sharp bound. It was also pointed that it is too difficult to prove it theoretically.

An effective approach is adaptive control. Adaptive control stability has been studied for a long time. It can be traced to the book \cite{Book1}. The core is to design updating laws. Early works on the adaptive control for synchronization of complex networks, readers can see \cite{Huang1}-\cite{Lu2} and others. For example, \cite{Zhou} proposed distributed
adaptive schemes for synchronization.


In this paper, we propose novel general adaptive algorithms, which are distributed and can be applied to synchronization, consensus of multi-agents and anti-synchronization of direct complex networks \cite{Chen1a}-\cite{book}.

\section{Some basic concepts and background}

Before giving main theoretical results, we need following result given in \cite{Chen1}.

\begin{propo}\label{p1}
Assume $L$ is connected, then for eigenvalue $0$, $[1,\dots,1]^{T}$ is the right eigenvector, $\theta=[\theta_{1},\cdots,\theta_{m}]^{T}\in
\mathbb{R}^{m}$ is the left eigenvector, $\theta_{i}>0$. In the following, we let
$\sum\limits_{i=1}^{m}\theta_{i}=1$.
\end{propo}

\begin{nota}
Define $\Theta=\mathrm{diag}\{\theta\}$, then the eigenvalues of $(\Theta L)^{s}=\frac{1}{2}(\Theta L+L^{T}\Theta)$ are sorted as: $0=\lambda_{1}>\lambda_{2}\ge\cdots\ge \lambda_{m}$.
\end{nota}

\begin{asu}\label{a1}
QUAD-Function class: Function $\mathscr{F}(\cdot)(\cdot)$ satisfies QUAD-condition, denoted as $\mathscr{F}(\cdot)\in (\Phi,P,\eta)$, if there exist positive definite matrix $P\in \mathbb{R}^{n\times n}$, matrix $\Phi\in \mathbb{R}^{n\times n}$, and scalar $\eta>0$, for any $z_1, z_2\in \mathbb{R}^{n}$,
\begin{align}\label{Q}
[z_1-z_2]^{T}P\big\{[\mathscr{F}(z_1)-\mathscr{F}(z_2)]-\Phi [z_1-z_2]\big\}
\le-\eta[z_1-z_2]^{T}[z_1-z_2].
\end{align}
\end{asu}

\begin{rem}
The concept of QUAD was first introduced in \cite{Chen1} for the case that $P$ and $\Phi$ are positive definite diagonal matrices. The QUAD condition introduced here is a natural generalization, which means that after some rotations, reflections or other transformations, not only enlargement, we can estimate the nonlinearity of $\mathscr{F}(\cdot)$ by matrices $P, \Phi$ and constant $\eta$. Readers can also refer to \cite{Lu1}.
\end{rem}

Here, we compare two definitions of QUAD-condition.

Let $P=Q^{T}JQ$ be its eigenvalue decomposition. Denote $\tilde{z}^i=Qz^i$, $\tilde{\mathscr{F}}(\tilde{z}^i)=Q\mathscr{F}(z^i)=Q\mathscr{F}(Q^{T}\tilde{z}^i)$, $\tilde{\Phi}=Q\Phi Q^{T}$. Then, (\ref{Q}) can be written as
\begin{align*}
[\tilde{z}^i-\tilde{z}^j]^{T}J\big\{[\tilde{\mathscr{F}}(\tilde{z}^i)
-\tilde{\mathscr{F}}(\tilde{z}^j)]
-\tilde{\Phi} [\tilde{z}^i-\tilde{z}^j]\big\}\le-\eta[\tilde{z}^i-\tilde{z}^j]^{T}
[\tilde{z}^i-\tilde{z}^j].
\end{align*}

If $\tilde{\Phi}$ is also a positive diagonal matrix, which is equivalent to that $P$ and $\Phi$ have same eigenvectors. In this case, $\tilde{\mathscr{F}}(\cdot)$ satisfies the QUAD-condition introduced in \cite{Chen1}, where $J$ and $\tilde{\Phi}$ are positive diagonal matrices.

\begin{nota}
The kernel of a matrix $A$ is defined as $ker(A)=\{u|Au=0,~u\in R^{n}\}$, while the orthogonal complement of $ker(A)$ is denoted by $ker(A)^{\bot}$.
\end{nota}

\begin{asu}\label{a2}
If Assumption \ref{a1} holds, $P\Gamma$ is semi-positive definite on $\mathbb{R}^n$, where $\Gamma$ is defined in (\ref{model1}), and $P\Gamma$ is positive definite on subspace ${ker(P\Phi)^{\bot}}$.
\end{asu}

\section{Global synchronization analysis for distributed adaptive algorithm of complex networks}

In this section, we discuss following adaptive algorithm model
\begin{align}\label{model1a}
\left\{\begin{array}{ll}
\dot{z}^{i}(t)=\mathscr{F}(z^{i}(t))+\sum_{j=1}^{m}w_{ij}(t)\Gamma
z^{j}(t),\\
\dot{w}_{ij}(t)=
\theta_{i}^{-1}[z^{i}(t)-z^{j}(t)]^{T}P\Gamma[z^{i}(t)-z^{j}(t)],~j\in \mathscr{N}(i),\\
w_{ij}(0)=l_{ij},
\end{array}\right.
\end{align}
where $\mathscr{N}(i)$ means the neighborhood of $i$, $w_{ij}(t)$ is the directed adaptive coupling weight at time $t>0$, $w_{ii}(t)=-\sum_{j\ne i}w_{ij}(t)$, and
$\dot{w}_{ii}(t)=-\sum_{j\in N(i)}\dot{w}_{ij}(t).$%

\begin{thm}\label{adap1}
Under Assumptions \ref{a1} and \ref{a2}, algorithm (\ref{model1a}) can reach synchronization.
\end{thm}

Firstly, we give two lemmas.
\begin{lem}\label{le1}
Under the assumption of Theorem \ref{adap1}, any $u\in \mathbb{R}^{n}$ can be decomposed to
$u=u_{1}+u_{2}$, where $u_{1}\in ker(P\Phi)^{\bot}$, $u_{2}\in ker(P\Phi)$, and $u_{1}\perp u_{2}$. Therefore, there is a constant $c_{1}$, such that
\begin{align}\label{c1}
u^{T}(P\Phi)u= u_{1}^{T}(P\Phi)u_{1}\le c_{1} u_{1}^{T}(P\Gamma)u_{1}\le c_{1} u^{T}(P\Gamma)u.
\end{align}
\end{lem}

\begin{lem}\label{le2}(see \cite{Chen2} equation (19))
If $A=(a_{ij})$ is a symmetric matrix satisfying $a_{ij}\ge 0$, if $i\ne j$
and $\sum_{i=1}^{m}a_{ij}=0$. Then, we have
\begin{align}\label{Basic}
\sum_{i,j=1}^{m}a_{ij}u^{{i}^{T}}v^{j} =-\frac{1}{2}\sum_{i,j=1}^ma_{ij}[u^{i}-u^{j}]^{T}[v^{i}-v^{j}].
\end{align}
\end{lem}

Lemma 2 was given in \cite{Chen2} to discuss synchronization of nonlinearly coupled networks. It plays a key role in discussing synchronization of the distributive adaptive algorithm.

\begin{lem}\label{le3}
Under algorithm (\ref{model1a}), for all $t>0$,
\begin{align*}
\sum_{i=1}^{m}\theta_{i}w_{ij}(t)=0, \sum_{i=1}^{m}\theta_{i}\dot{w}_{ij}(t)=0,
~\forall j; \sum_{j=1}^{m}w_{ij}(t)=0,\sum_{j=1}^{m}\dot{w}_{ij}(t)=0,~\forall i.
\end{align*}
\end{lem}

\emph{Proof of Theorem \ref{adap1}:} Define
\begin{align*}
V(t)=\frac{1}{2}\sum_{i=1}^{m}\theta_{i}(z^{i}(t)-\bar{z}(t))^{T}
P(z^{i}(t)-\bar{z}(t))+\frac{1}{4}\sum_{i\ne j}(\theta_{i}w_{ij}(t)-c\theta_{i}l_{ij})^{2},
\end{align*}
where $\bar{z}(t)=\sum\limits_{i=1}^{m}\theta_{i}z_{i}(t)$.

Differentiating it, one can get
\begin{align*}
\dot{V}(t)
=&\sum_{i=1}^{m}\theta_{i}[z^{i}(t)-\bar{z}(t)]^{T}P[\mathscr{F}(z^{i}(t))-\mathscr{F}(\bar{z}(t))]\nonumber\\
&+\sum_{i=1}^{m}[z^{i}(t)-\bar{z}(t)]^{T}
\sum_{j=1}^{m}\theta_{i}w_{ij}(t)P\Gamma[z^{j}(t)-\bar{z}(t)]
\nonumber\\
&-\frac{1}{2}\sum_{i\ne j}
[c\theta_{i}l_{ij}-\theta_{i}w_{ij}(t)][z^{i}(t)-z^{j}(t)]^{T}P\Gamma[z^{i}(t)-z^{j}(t)].
\end{align*}
By QUAD-condition and Lemma \ref{le1}, we have
\begin{align*}
&\sum_{i=1}^{m}\theta_{i}[z^{i}(t)-\bar{z}(t)]^{T}
P[\mathscr{F}(z^{i}(t))-\mathscr{F}(\bar{z}(t))]
\nonumber\\
\le & \sum_{i=1}^{m}\theta_{i}[z^{i}(t)-\bar{z}(t)]^{T}P\Phi[z^{i}(t)-\bar{z}(t)]-\eta \sum_{i=1}^{m}[z^{i}(t)-\bar{z}(t)]^{T}[z^{i}(t)-\bar{z}(t)]\nonumber\\
\le& c_{1}\sum_{i=1}^{m}\theta_{i}[z^{i}(t)-\bar{z}(t)]^{T}
P\Gamma[z^{i}(t)-\bar{z}(t)]-\eta \sum_{i=1}^{m}[z^{i}(t)-\bar{z}(t)]^{T}[z^{i}(t)-\bar{z}(t)].
\end{align*}

Based on Lemma \ref{le2} and Lemma \ref{le3}, we have
\begin{align*}
&2\sum_{i=1}^{m}[z^{i}(t)-\bar{z}(t)]^{T}
\sum_{j=1}^{m}\theta_{i}l_{ij}P\Gamma[z^{j}(t)-\bar{z}(t)]\\
=&\sum_{i,j=1}^{m}[z^{i}(t)-\bar{z}(t)]^{T}
(\theta_{i}l_{ij}+\theta_{j}l_{ji})P\Gamma[z^{j}(t)-\bar{z}(t)]\\
=&-\frac{1}{2}\sum_{i,j=1}^m[z^{i}(t)-z^{j}(t)]^{T}
(\theta_{i}l_{ij}+\theta_{j}l_{ji})P\Gamma[z^{i}(t)-z^{j}(t)]\\
=&-\sum_{i,j=1}^m[z^{i}(t)-z^{j}(t)]^{T}
\theta_{i}l_{ij}P\Gamma[z^{i}(t)-z^{j}(t)],
\end{align*}
and
\begin{align*}
&\sum_{i=1}^{m}[z^{i}(t)-\bar{z}(t)]^{T}
\sum_{j=1}^{m}\theta_{i}w_{ij}(t)P\Gamma[z^{j}(t)-\bar{z}(t)]\\
=&-\frac{1}{2}\sum_{i,j=1}^m[z^{i}(t)-z^{j}(t)]^{T}
\theta_{i}w_{ij}(t)P\Gamma[z^{i}(t)-z^{j}(t)].
\end{align*}
Therefore, we have
\begin{align*}
\dot{V}(t)
\le& c_{1}\sum_{i=1}^{m}\theta_{i}[z^{i}(t)-\bar{z}(t)]^{T}P
\Gamma[z^{i}(t)-\bar{z}(t)]\\
&-\eta \sum_{i=1}^{m}[z^{i}(t)-\bar{z}(t)]^{T}[z^{i}(t)-\bar{z}(t)]\\
&+\frac{c}{2}\sum_{i,l=1}^{m}[z^{i}(t)-\bar{z}(t)]^{T}
(\theta_{i}l_{ij}+\theta_{j}l_{ji})P\Gamma[z^{i}(t)-\bar{z}(t)].
\end{align*}

Since
\begin{align*}
&\frac{c}{2}\sum_{i=1}^{m}(z^{i}(t)-\bar{z}(t))^{T}
\sum_{j=1}^{m}(\theta_{i}l_{ij}+\theta_{j}l_{ji})P\Gamma(z^{j}(t)-\bar{z}(t))\\
\le& \frac{c\lambda_{2}}{\max_i\theta_i}\sum_{i,j=1}^{m}\theta_i(z^{i}(t)-\bar{z}(t))^{T}P\Gamma
(z^{j}(t)-\bar{z}(t)).
\end{align*}

In case that $c>c_{1}\max_i\theta_i|\lambda_{2}|^{-1}$, we have
%
\begin{align}\label{zj}
\dot{V}(t)\le -\eta\sum_{i=1}^{m}\theta_{i}[z^{i}(t)-\bar{z}(t)]^{T}[z^{i}(t)-\bar{z}(t)],
\end{align}
and
\begin{align*}
\int_{0}^{t}
\sum_{i=1}^{m}\theta_{i}[z^{i}(\kappa)-\bar{z}(\kappa)]^{T}[z^{i}(\kappa)-\bar{z}(\kappa)]d\kappa<\frac{1}{\eta}[V(0)-V(t)].
\end{align*}
Therefore, $\lim_{t\rightarrow \infty}\sum_{i=1}^{m}\theta_{i}[z^{i}(t)-\bar{z}(t)]^{T}[z^{i}(t)-\bar{z}(t)]=0$.
Theorem is proved completely.

\begin{rem}
It has been explored that the left and right eigenvectors corresponding to the eigenvalue $0$ of the coupling matrix plays key roles in discussing synchronization. Therefore, to keep the left and right eigenvectors corresponding to the eigenvalue $0$ of all coupling matrices $W(t)=[w_{ij}(t)]$, $t>0$, invariant is the most important. Our algorithm  ensures that all coupling matrices $W(t)=[w_{ij}(t)]$ have same left and right eigenvectors corresponding to the eigenvalue $0$.
\end{rem}

\begin{rem}
It is clear that QUAD condition and the equality (5) in Lemma 2 (see \cite{Chen2} equation (19)) play key role in proof of Theorem.
\end{rem}

\begin{col}
Assume $L$ is symmetric or node balanced. Thus all $\theta_{i}=\frac{1}{m}$, $i=1,\cdots,m$. Under the assumptions made in Theorem \ref{adap1}, the following algorithm
\begin{align*}
\left\{\begin{array}{ll}
\dot{x}^{i}(t)=\mathscr{F}(z^{i}(t))+\sum_{j=1}^{m}w_{ij}(t)\Gamma z^{j}(t),\\
\dot{w}_{ij}(t)=
[z^{i}(t)-z^{j}(t)]^{T}P\Gamma[z^{i}(t)-z^{j}(t)],~j\in N(i),\\
w_{ij}(0)=l_{ij},~i,j=1,\cdots,m
\end{array}\right.
\end{align*}
can reach synchronization.
\end{col}

\section{Adaptive pinning synchronization of complex networks with a single controller}

We discuss adaptive synchronization with a single pinning controller, which is identical to synchronization of leader-follower systems, see \cite{Chen2}.

Consider the following pinning control model \cite{Chen2},
\begin{align}\label{pin}
\left\{\begin{array}{l}
\dot{z}^1(t)=\mathscr{F}(z^1(t))+c\sum\limits_{j=1}^m l_{1j}\Gamma z^j(t)
-c\varepsilon\Gamma[z^{1}(t)-s(t)],\\
\dot{z}^i(t)=\mathscr{F}(z^i(t))+c\sum\limits_{j=1}^ml_{ij}\Gamma z^j(t),
i=2,\cdots,m\end{array}\right.
\end{align}
where $s(t)$ be a solution of $\dot{s}(t)=\mathscr{F}(s(t))$. Its adaptive algorithm
\begin{align}\label{adappin1}
\left\{
\begin{array}{ll}
\dot{z}^1(t)=&\mathscr{F}(z^1(t))
+\sum\limits_{j=1}^m w_{1j}(t)\Gamma z^j(t)-w_{10}(t)\Gamma[z^{1}(t)-s(t)],\\
\dot{z}^i(t)=&\mathscr{F}(z^i(t))+\sum\limits_{j=1}^m w_{ij}(t)\Gamma z^j(t),~i=2,\cdots,m\\
\dot{w}_{ij}(t)=&\theta_{i}^{-1}
[z^{i}(t)-z^j(t)]^{T}P\Gamma[z^{i}(t)-z^j(t)],~~j\in N(i),
\\
\dot{w}_{10}(t)=&\theta_{1}^{-1}[z^{1}(t)-s(t))]^{T}P\Gamma[z^{1}(t)-s(t)],\\
w_{ij}(0)=&l_{ij},~i,j=1,\cdots,m, w_{10}(0)=\varepsilon
\end{array}
\right.
\end{align}
Similarly, we have
\begin{thm}\label{adappint}
Under Assumptions \ref{a1} and \ref{a2}, distributive adaptive algorithm (\ref{adappin1}) can synchronize to $s(t)$.
\end{thm}

\begin{proof}
Define the following candidate Lyapunov function
\begin{align*}
V(t)=&\frac{1}{2}\sum_{i=1}^{m}\theta_{i}[z^{i}(t)-s(t)]^{T}
P[z^{i}(t)-s(t)]\\
&+\frac{1}{4}\sum_{i\ne j}^{m}[\theta_{i}w_{ij}(t)-c\theta_{i}l_{ij}]^{2}
+\frac{1}{2}[\theta_{1}w_{10}(t)-c\theta_{1}\varepsilon]^{2},
\end{align*}
where $c$ will be given later.

Differentiating $V(t)$, one can get
\begin{align*}
\dot{V}(t)=&\sum_{i=1}^{m}\theta_{i}[z^{i}(t)-s(t)]^{T}P[\mathscr{F}(z^{i}(t))-\mathscr{F}(s(t))]\\
&+\sum_{i=1}^{m}[z^{i}(t)-s(t)]^{T}
\sum_{j=1}^{m}\theta_{i}w_{ij}(t)P\Gamma[z^{j}(t)-s(t)]\\
&-\theta_{1}w_{10}(t)[z^{1}(t)-s(t)]^{T}P\Gamma[z^{1}(t)-s(t)]\\
&-\frac{1}{2}\sum_{i,j=1}^{m}[c\theta_{i}l_{ij}-\theta_{i}w_{ij}(t)][z^{i}(t)-z^{j}(t)]^{T}P\Gamma[z^{i}(t)-z^{j}(t)]
\\&+\theta_{1}[w_{10}(t)-c\varepsilon][z^{1}(t)-s(t)]^{T}P\Gamma[z^{1}(t)-s(t)]\\
=&\sum_{i=1}^{m}\theta_{i}[z^{i}(t)-s(t)]^{T}P[\mathscr{F}(z^{i}(t))-\mathscr{F}(s(t))]\\
&-c\theta_{1}\varepsilon[z^{1}(t)-s(t)]^{T}P\Gamma[z^{1}(t)-s(t)]\\
&-\frac{c}{2}\sum_{i,j=1}^{m}\theta_{i}l_{ij}[z^{i}(t)-z^{j}(t)]^{T}P\Gamma[z^{i}(t)-z^{j}(t)].
\end{align*}
By Lemma \ref{le2}, we have
\begin{align*}
&-\frac{c}{2}\sum_{i,j=1}^{m}\theta_{i}l_{ij}[z^{i}(t)-z^{j}(t)]^{T}
P\Gamma[z^{i}(t)-z^{j}(t)]\\
=&c\sum_{i,j=1}^{m}\theta_{i}l_{ij}[z^{i}(t)-s(t)]^{T}P\Gamma[z^{i}(t)-s(t)].
\end{align*}
By similar derivations given in Theorem \ref{adap1}, we have
\begin{align*}
\dot{V}(t)
\le& -\eta\sum_{i=1}^{m}\theta_{i}[z^{i}(t)-s(t)]^{T}[z^{i}(t)-s(t)]\\
&+c_{1}\sum_{i=1}^{m}\theta_{i}[z^{i}(t)-s(t)]^{T}P\Gamma[z^{i}(t)-s(t)]
\nonumber\\
&+c\sum_{i=1}^{m}[z^{i}(t)-s(t)]^{T}
\sum_{j=1}^{m}\theta_{i}\tilde{l}_{ij}P\Gamma[z^{j}(t)-s(t)],
\end{align*}
where $c_{1}$ is the constant given in (\ref{c1}) in Lemma \ref{le1}, $\tilde{l}_{ij}=l_{ij}$, except $\bar{l}_{11}=l_{11}-\epsilon$. Define a matrix $\tilde{L}=(\tilde{l}_{ij})_{i,j=1}^{m}$. \cite{Chen2} pointed out that all eigenvalues $\mu_{i}$, $i=1,\cdots,m$ of $(\Theta \tilde{L})^{s}=\frac{1}{2}[\Theta \tilde{L}+\tilde{L}^{T}\Theta ]$ satisfy $0>\mu_{1}\ge \mu_{2}\ge\cdots\ge \mu_{m}$. If $c>c_{1}\max_i\theta_{i}|\mu_{1}|^{-1}$, we have
\begin{align*}
&c_{1}\sum_{i=1}^{m}\theta_{i}[z^{i}(t)-s(t)]^{T}P\Gamma[z^{i}(t)-s(t)]\\
+&cd_{1}\sum_{i=1}^{m}[z^{i}(t)-s(t)]^{T}
\sum_{j=1}^{m}\theta_{i}\tilde{l}_{ij}P\Gamma[z^{j}(t)-s(t)]<0,
\end{align*}
which implies
\begin{align*}
\dot{V}(t)\le -\eta\sum_{i=1}^{m}\theta_{i}[z^{i}(t)-s(t)]^{T}[z^{i}(t)-s(t)],
\end{align*}
and
\begin{align*}
\int_{0}^{t}
\sum_{i=1}^{m}\theta_{i}[z^{i}(\kappa)-s(\kappa)]^{T}[z^{i}(\kappa)-s(\kappa)]d\kappa<\frac{1}{\eta}[V(0)-V(t)].
\end{align*}
Therefore,
$\lim_{t\rightarrow \infty}\sum_{i=1}^{m}\theta_{i}[z^{i}(t)-s(t)]^{T}[z^{i}(t)-s(t)]=0$.
Theorem is proved completely.
\end{proof}

\section{Applications to Consensus of Multi-agents Systems}
As applications, we discuss consensus of multi-agents,
\begin{align}
\frac{dz^i(t)}{dt}=Az^i(t)+c\sum\limits_{j=1}^ml_{ij}\Gamma z^j(t),~i=1,\cdots,m
\end{align}

Consider the model
\begin{align}\label{con}
\dot{z}(t)=Az(t)+Bu(t)
\end{align}
If the system (\ref{con}) is controllable, there are positive definite matrix $P$, such that
\begin{align}
PA+A^{T}P-2BB^{T}<0,
\end{align}
which is equivalent to the following QUAD condition
\begin{align*}
[z_1-z_2]^{T}P\big\{[Az_1-Az_2]-P^{-1}BB^{T}[z_1-z_2]\big\}\le -\eta
[z_1-z_2]^{T}[z_1-z_2].
\end{align*}

By the results obtained in previous section, we can give
\begin{thm}\label{adapcontrol}
If the system (\ref{con}) is controllable, the distributive adaptive system
\begin{align*}
\left\{\begin{array}{ll}
\dot{z}^i(t)=&Az^i(t)
+\sum\limits_{j=1}^m w_{1j}(t)P^{-1}BB^{T} z^j(t)\\
\dot{w}_{ij}(t)=&\theta_{i}^{-1}
[z^{i}(t)-z^{j}(t)]^{T}BB^{T}[z^{i}(t)-z^{j}(t)],~~j\in \mathscr{N}(i),
\\
w_{ij}(0)=&l_{ij},
\end{array}\right.
\end{align*}
can reach consensus.

If $\dot{s}(t)=As(t)$, then following adaptive algorithm
\begin{align*}
\left\{\begin{array}{l}
\dot{z}^1(t)=Az^1(t)
+\sum\limits_{j=1}^m w_{1j}(t)P^{-1}BB^{T} z^j(t)-w_{10}(t) P^{-1}BB^{T}[z^{1}(t)-s(t)],\\
\dot{z}^i(t)=Az^i(t)+\sum\limits_{j=1}^m w_{ij}(t)P^{-1}BB^{T} z^j(t),~i\ne 1\\
\dot{w}_{ij}(t)=\theta_{i}^{-1}
[z^{i}(t)-s(t)]^{T}BB^{T}[z^{i}(t)-s(t)],~~j\in N(i),\\
\dot{w}_{10}(t)=\theta_{1}^{-1}[z^{1}(t)-s(t))]^{T}BB^{T}[z^{1}(t)-s(t)],\\
w_{ij}(0)=l_{ij}, w_{10}(0)=\varepsilon
\end{array}\right.
\end{align*}
can reach $z^i(t)-s(t)\rightarrow 0$, for all $i=1,\cdots,m$.
\end{thm}

Let $\Gamma=P^{-1}BB^{T}$. Theorem is a direct consequence of Theorem \ref{adap1} and Theorem 2.

\section{Anti-synchronization}
Next, we discuss anti-synchronization \cite{Anti1,Anti2,Wang,Chen2019},
\begin{align}\label{anti0a}
\dot{z}^{i}(t)=\mathscr{F}(z^{i}(t))
+c\sum_{j\ne i}^{m}|l_{ij}|\Gamma[\mathrm{sign}(l_{ij})z^{j}(t)-z^{i}(t)],
\end{align}
where all nodes are connected and can be split into two subgroups $\mathcal{V}_{1}$ and $\mathcal{V}_{2}$,  such that $l_{ij}\ge 0$ for all $i,j\in \mathcal{V}_{p}$ ($p=1,2$),  $l_{ij}\le 0$ for all $i\in \mathcal{V}_{p},j\in \mathcal{V}_{q}$, $p\ne q$.

Let $\hat{z}_{i}(t)=z^{i}(t)$, if $i\in \mathcal{V}_{1}$, and $\hat{z}^{i}(t)=-z_{i}(t)$, if $i\in \mathcal{V}_{2}$,
then (\ref{anti0a}) can be rewritten as
\begin{align}\label{anti1}
\left\{ \begin{array}{ll}
\dot{\hat{z}}^{i}(t)=\mathscr{F}(\hat{z}^{i}(t))
+c\sum_{j=1}^{m}l_{ij}^{*}\Gamma{\hat{z}}^{j}(t),~i\in \mathcal{V}_{1},~j\in \mathscr{N}(i), \\
\dot{\hat{z}}^{i}(t)=-\mathscr{F}(-\hat{z}^{i}(t))
+c\sum_{j=1}^{m}l_{ij}^{*}\Gamma{\hat{z}}^{j}(t),~i\in \mathcal{V}_{2},~j\in \mathscr{N}(i),
\end{array}
\right.
\end{align}
where $l_{ij}^{*}=|l_{ij}|$, if $i\ne j$, and $l_{ii}^{*}=-\sum_{j\ne i}^{m}|l_{ij}|$.
Therefore, anti-synchronization of (\ref{anti0a}) is equivalent to synchronization of $\hat{z}^{i}(t)$ for system (\ref{anti1}).

For matrix $L^{*}=(l_{ij}^{*})$, $\theta=[\theta_{1},\cdots,\theta_{m}]^{T}$ is the left eigenvector in Proposition \ref{p1}.

\begin{asu}\label{a3} For $\mathscr{F}(\cdot)$, suppose
\begin{align}\label{QA}
\left\{
\begin{array}{l}
(z_1-z_2)^{T}Q\big\{[\mathscr{F}(z_1)-\mathscr{F}(z_2)]-\Phi (z_1-z_2)\big\}\le -\eta(z_1-z_2)^{T}(z_1-z_2), \\
(z_1+z_2)^{T}Q\big\{[\mathscr{F}(z_1)+\mathscr{F}(z_2)]-\Phi (z_1+z_2)\big\}\le -\eta(z_1+z_2)^{T}(z_1+z_2),
\end{array}
\right.
\end{align}
where $Q$ is a positive definite matrix.
\end{asu}

In \cite{Chen2019}, it was proved that under Assumption \ref{a3}, (\ref{anti1}) can reach synchronization if $c$ is sufficiently large.

Based on previous discussion, the adaptive algorithm takes following form
\begin{align}
\left\{\begin{array}{ll}
\dot{\hat{z}}^{i}(t)=\mathscr{F}(\hat{z}^{i}(t))+\sum_{j=1}^{m}w_{ij}(t)\Gamma[
\hat{z}^{j}(t)-\hat{z}^{i}(t)],\\
\dot{w}_{ij}(t)=
\theta_{i}^{-1}[\hat{z}^{i}(t)-\hat{z}^{j}(t)]^{T}P\Gamma[\hat{z}^{i}(t)-\hat{z}^{j}(t)],~~
if~i\in \mathcal{V}_{1},j\in \mathscr{N}(i),\\
\dot{\hat{z}}^{i}(t)=-\mathscr{F}(-\hat{z}^{i}(t))+\sum_{j=1}^{m}w_{ij}(t)\Gamma[
\hat{z}^{j}(t)-\hat{z}^{i}(t)],\\
\dot{w}_{ij}(t)=\theta_{i}^{-1}[\hat{z}^{i}(t)-\hat{z}^{j}(t)]^{T}P\Gamma[\hat{z}^{i}(t)-\hat{z}^{j}(t)],~~
if~i\in \mathcal{V}_{2}, j\in \mathscr{N}(i),\\
{w}_{ij}(0)=l_{ij}^{*},~i,j=1,\cdots,m
\end{array}\right.\label{adaptiveanti1}
\end{align}

\begin{thm}
Under Assumptions \ref{a2} and \ref{a3}, anti-synchronization can be reached for adaptive system (\ref{adaptiveanti1}).
\end{thm}

\proof
Define anti-synchronization state as
\begin{align}
\bar{\hat{z}}(t)=\sum_{i=1}^{m}\theta_{i}\hat{z}^{i}(t)=\sum\limits_{j\in \mathcal{V}_1}\theta_{j}z^{j}(t)-\sum_{j\in \mathcal{V}_2}\theta_{i}z^{j}(t),
\end{align}
and Lyapunov function takes form
\begin{align*}
V(t)=\frac{1}{2}\sum_{i=1}^{m}\theta_{i}[\hat{z}^{i}(t)-\bar{\hat{z}}(t)]^{T}
P[\hat{z}^{i}(t)-\bar{\hat{z}}(t)]+\frac{1}{4}\sum_{i,j=1}^{m}[\theta_{i}w_{ij}(t)-c\theta_{i}l_{ij}^{*}]^{2}.
\end{align*}

For $i\in \mathcal{V}_1$, from (\ref{QA}),
\begin{align*}
&[\hat{z}^{i}(t)-\bar{\hat{z}}(t)]^{T}P
\big[\mathscr{F}(\hat{z}^{i}(t))-\mathscr{F}(\bar{\hat{z}}(t))\big]\\
\le& [\hat{z}^{i}(t)-\bar{\hat{z}}(t)]^{T}P\Phi[\hat{z}^{i}(t)-\bar{\hat{z}}(t)]-\eta[\hat{z}^{i}(t)-\bar{\hat{z}}(t)]^{T}[\hat{z}^{i}(t)-\bar{\hat{z}}(t)].
\end{align*}

For $i\in \mathcal{V}_2$, from (\ref{QA}),
\begin{align*}
&[\hat{z}^{i}(t)-\bar{\hat{z}}(t)]^{T}P\big[-\mathscr{F}(-\hat{z}^{i}(t))-\mathscr{F}(\bar{\hat{z}}(t))\big]\\
\le& [\hat{z}^{i}(t)-\bar{\hat{z}}(t)]^{T}P\Phi[\hat{z}^{i}(t)-\bar{\hat{z}}(t)]-\eta[\hat{z}^{i}(t)-\bar{\hat{z}}(t)]^{T}[\hat{z}^{i}(t)-\bar{\hat{z}}(t)].
\end{align*}

By the same derivations as before, we can prove
$\dot{V}(t)\le -\eta\sum_{i=1}^{m}\theta_{i}[\hat{z}^{i}(t)-\bar{\hat{z}}(t)]^{T}
[\hat{z}^{i}(t)-\bar{\hat{z}}(t)]$,
and
\begin{align*}
\int_{0}^{t}
\sum_{i=1}^{m}\theta_{i}[\hat{z}^{i}(\kappa)-\bar{\hat{z}}(\kappa)]^{T}
[\hat{z}^{i}(\kappa)-\bar{\hat{z}}(\kappa)]d\kappa<\frac{1}{\eta}[V(0)-V(t)].
\end{align*}
Therefore, $\lim_{t\rightarrow \infty}\sum_{i=1}^{m}\theta_{i}[\hat{z}^{i}(t)-\bar{\hat{z}}(t)]^{T}[z^{i}(t)
-\bar{\hat{z}}(t)]=0$, anti-synchronization is realized.

\begin{rem}
In case $\mathscr{F}(-x)=-\mathscr{F}(x)$, such as $\mathscr{F}(x)=Ax$, (\ref{Q}) and (\ref{QA}) will be identical. The adaptive algorithm (\ref{adaptiveanti1}) becomes
\begin{align*}
\left\{\begin{array}{ll}
\dot{\hat{z}}^{i}(t)=\mathscr{F}(\hat{z}^{i}(t))+\sum_{j=1}^{m}w_{ij}(t)\Gamma[
\hat{z}^{j}(t)-\hat{z}^{i}(t)],\\
\dot{w}_{ij}(t)=
\theta_{i}^{-1}[\hat{z}^{i}(t)-\hat{z}^{j}(t)]^{T}
P\Gamma[\hat{z}^{i}(t)-\hat{z}^{j}(t)], ~j\in \mathscr{N}(i),\\
{w}_{ij}(0)=l_{ij}^{*},~i,j=1,\cdots,m
\end{array}\right.
\end{align*}

\end{rem}

\section{Global synchronization analysis for distributed adaptive algorithm of nonlinear coupled complex networks}
In practice, we can not observe the state $z^{i}(t)$
directly. Instead, we only can obtain data
$\mathscr{G}(z^{i}(t))$, $i=1,\cdots,m$.
We need to use $\mathscr{G}(z^{i}(t))$ to synchronize the
uncoupled system, which means that the synchronization scheme is
nonlinear,
\begin{align*}
\frac{d z^{i}(t)}{dt}=\mathscr{F}(z^{i}(t),t)+\sum\limits_{j=1,j\neq
i}^{m}a_{ij} [\mathscr{G}(z^{j}(t))-\mathscr{G}(z^{i}(t))],
\end{align*}
where $\mathscr{G}(z)=[g_{1}(z_{1}),\cdots,g_{n}(z_{n})]^T$, for $z=[z_1,\cdots,z_n]^T$, and
satisfying $[g_{i}(u)-g_{i}(v)]>\beta [u-v]$, $i=\onetom$, $\beta>0$, for any $u\ne v$.

\begin{thm}
Suppose matrix $L$ is symmetric, the QUAD-condition $(I_n, I_n, \eta)$ holds, the following algorithm
\begin{align}\label{adaptive}
\left\{\begin{array}{l}
\dot{x}^{i}(t)=\mathscr{F}(z^{i}(t))+\sum_{j=1}^{m}w_{ij}(t)
[\mathscr{G}(z^{j}(t))-\mathscr{G}(z^{i}(t))];\\
\dot{w}_{ij}(t)=
\rho[\mathscr{G}(z^{i}(t))-\mathscr{G}(z^{j}(t)]^{T}[\mathscr{G}(z^{i}(t))-\mathscr{G}(z^{j}(t))],~~j\in \mathscr{N}(i);\\
w_{ij}(0)=l_{ij}
\end{array}\right.
\end{align}
can realize the synchronization, where $\rho\beta<1$.
\end{thm}

Clearly, $\dot{w}_{ij}(t)=\dot{w}_{ji}(t)$, which implies that $(w_{ij}(t))$ is symmetric for all $t>0$.

\begin{proof}
Define
\begin{align*}
V(t)=&\frac{1}{2}\sum_{i=1}^{m}[z^{i}(t)-\bar{z}(t)]^{T}
[z^{i}(t)-\bar{z}(t)]+\frac{1}{4}\sum_{i=1}^{m}[w_{ij}(t)-cl_{ij}]^{2}.
\end{align*}
Based on the equality (\ref{Basic}), we have
\begin{align*}
&\sum_{i=1}^{m}[z^{i}(t)-\bar{z}(t)]^{T}
\sum_{j=1}^{m}w_{ij}(t)[\mathscr{G}(z^{j}(t))-\mathscr{G}(\bar{z}(t))]\\
=&-\frac{1}{2}\sum_{i,j=1}^{m}[z^{i}(t)-z^{j}(t)]^{T}w_{ij}(t)[\mathscr{G}(z^{i}(t))-\mathscr{G}(z^{j}(t))],
\end{align*}
and differentiate it,
\begin{align*}
\dot{V}(t)
=&\sum_{i=1}^{m}[z^{i}(t)-\bar{z}(t)]^{T}[\mathscr{F}(z^{i}(t))-\mathscr{F}(\bar{z}(t))]\\
&+\sum_{i,j=1}^{m}[z^{i}(t)-\bar{z}(t)]^{T}
w_{ij}(t)[\mathscr{G}(z^{j}(t))-\mathscr{G}(\bar{z}(t))]
\nonumber\\
&+\frac{\rho}{2}\sum_{i,j=1}^{m}
w_{ij}(t)[\mathscr{G}(z^{i}(t))-\mathscr{G}(z^{j}(t))]^{T}
[\mathscr{G}(z^{i}(t))-\mathscr{G}(z^{j}(t))]
\nonumber\\
&-\frac{c\rho}{2}\sum_{i\ne j}
l_{ij}[\mathscr{G}(z^{i}(t))-\mathscr{G}(z^{j}(t)]^{T}[\mathscr{G}(z^{i}(t))-\mathscr{G}(z^{j}(t))]
\nonumber\\
=&I_{1}(t)+I_{2}(t)+I_{3}(t)+I_{4}(t).
\end{align*}
By (\ref{Q}), we have
\begin{align*}
I_{1}(t)
&\le \sum_{i=1}^{m}[z^{i}(t)-\bar{z}(t)]^{T}
[z^{i}(t)-\bar{z}(t)]-\eta \sum_{i=1}^{m}[z^{i}(t)-\bar{z}(t)]^{T}[z^{i}(t)-\bar{z}(t)],\\
I_{2}(t)
&=-\frac{1}{2}\sum_{i,j=1}^{m}[z^{i}(t)-z^{j}(t)]^{T}
w_{ij}(t)[\mathscr{G}(z^{i}(t))-\mathscr{G}(z^{j}(t))]
\nonumber\\
&\le -\frac{\beta}{2}\sum_{i,j=1}^{m}[z^{i}(t)-z^{j}(t)]^{T}
w_{ij}(t)[z^{i}(t)-z^{j}(t)]\\
&=\beta\sum_{i,j=1}^{m}[z^{i}(t)-\bar{z}(t)]^{T}
w_{ij}(t)[z^{i}(t)-\bar{z}(t)],\\
I_{3}(t)
&\le -\rho\beta^{2}\sum_{i,j=1}^{m}[z^{i}(t)-\bar{z}(t)]^{T}
w_{ij}(t)[z^{i}(t)-\bar{z}(t)],\\
I_{4}(t)
&\le c\rho\beta^{2}\sum_{i,j=1}^{m}[z^{i}(t)-\bar{z}(t)]^{T}
l_{ij}[z^{i}(t)-\bar{z}(t)].
\end{align*}
In case $\rho\beta<1$, we have
$I_{2}(t)+I_{3}(t)<0$. Pick a sufficiently large $c$, we can prove
$$\dot{V}(t)<I_{1}(t)+I_{4}(t)<-\eta \sum_{i=1}^{m}[z^{i}(t)-\bar{z}(t)]^{T}[z^{i}(t)-\bar{z}(t)].$$
The rest follows the proof of Theorem \ref{adap1}.
\end{proof}

\section{Conclusions}
In this paper, we discuss the distributive adaptive synchronization of complex networks, consensus of multi-agents with or without pinning controller. New distributive adaptive algorithms are proposed. It is explored that the left eigenvector $\theta=[\theta_{1},\cdots,\theta_{m}]^{T}$ corresponding to eigenvalue $0$ of the coupling matrix $L$ plays a key role. As direct consequences, distributive adaptive algorithms for consensus, anti-synchronization with or without pinning controller are given, too. Of course, the theoretical results obtained in this paper can be applied potentially in real world complex networks, such as genetic regulatory networks, sensor networks, etc.

\end{document}